\documentclass[pdftex,leqno,fleqn,11pt]{article}

\usepackage[latin1]{inputenc}   
\usepackage{xspace}

\usepackage{graphicx}           
\usepackage{subfigure}

\usepackage{latexsym}           
\usepackage{amsmath}
\usepackage{amssymb}
\usepackage{amsfonts}

\usepackage[sort&compress,numbers]{natbib}

\usepackage{fullpage}

\usepackage[ruled,linesnumbered]{algorithm2e}

\usepackage{theorem}            

\newcommand{\rep}{\mbox{rep}}
\newcommand{\cl}{\mbox{cl}}
\newcommand{\dist}{\mbox{dist}}
\newcommand{\etal}{\textit{et al.}}

\newtheorem{theorem}{Theorem}[section]

\newtheorem{lemma}[theorem]{Lemma}

\newtheorem{definition}[theorem]{Definition}

\newtheorem{problem}[theorem]{Problem}

\newenvironment{proof}{{\bf Proof:} \rm}{\hfill $\square$ \medskip\\}

\title{Spanners of Complete $k$-Partite Geometric Graphs} 
\author{Prosenjit Bose\thanks{School of Computer Science,
    Carleton University, Ottawa, Ontario, Canada K1S 5B6.
    Research partially supported by NSERC, MRI, CFI, and 
    MITACS.}
\and Paz Carmi\footnotemark[1] 
\and Mathieu Couture\footnotemark[1]
\and Anil Maheshwari\footnotemark[1] 
\and Pat Morin\footnotemark[1] 
\and Michiel Smid\footnotemark[1]
}

\begin{document}

\maketitle

\begin{abstract}
We address the following problem: Given a complete $k$-partite geometric 
graph $K$ whose vertex set is a set of $n$ points in $\mathbb{R}^d$,
compute a spanner of $K$ that has a ``small'' stretch factor and 
``few'' edges. We present two algorithms for this problem. The first 
algorithm computes a $(5+\epsilon)$-spanner of $K$ with $O(n)$ edges 
in $O(n \log n)$ time.  The second algorithm computes a 
$(3+\epsilon)$-spanner of $K$ with $O(n \log n)$ edges in $O(n \log n)$ 
time. The latter result is optimal: We show that for any 
$2 \leq k \leq n - \Theta(\sqrt{n \log n} )$, spanners with 
$O(n \log n)$ edges and stretch factor less than $3$ do not exist 
for all complete $k$-partite geometric graphs. 
\end{abstract}


\section{Introduction}
Let $S$ be a set of $n$ points in $\mathbb{R}^d$. 
A \emph{geometric graph} with vertex set $S$ is an undirected graph $H$  
whose edges are line segments $\overline{pq}$ that are weighted by the 
Euclidean distance $|pq|$ between $p$ and $q$. For any two points $p$ and 
$q$ in $S$, we denote by $\delta_H(p,q)$ the length of a shortest path 
in $H$ between $p$ and $q$. For a real number $t \geq 1$, a subgraph 
$G$ of $H$ is said to be a $t$-\emph{spanner} of $H$, if 
$\delta_G(p,q) \leq t \cdot \delta_H(p,q)$ for all points $p$ and 
$q$ in $S$. The smallest $t$ for which this property holds is called the 
\emph{stretch factor} of $G$. Thus, a subgraph $G$ of $H$ with stretch 
factor $t$ approximates the $n \choose 2$ pairwise shortest-path lengths 
in $H$ within a factor of $t$. If $H$ is the complete geometric graph 
with vertex set $S$, then $G$ is also called a $t$-spanner of the point 
set $S$. 

Most of the work on constructing spanners has been done for the case 
when $H$ is the complete graph. It is well known that for any set $S$ 
of $n$ points in $\mathbb{R}^d$ and for any real constant $\epsilon > 0$, 
there exists a $(1+\epsilon)$-spanner of $S$ containing $O(n)$ edges. 
Moreover, such spanners can be computed in $O(n \log n)$ time; see 
Salowe~\cite{s-cmsg-91} and Vaidya~\cite{v-sgagc-91}. For a detailed 
overview of results on spanners for point sets, see the book by
Narasimhan and Smid~\cite{smid07}.

For spanners of arbitrary geometric graphs, much less is known. 
Alth{\"o}fer \emph{et al.}~\cite{addjs-sswg-93} have shown that for any
$t>1$, every weighted graph $H$ with $n$ vertices contains a subgraph
with $O(n^{1+2/(t-1)})$ edges, which is a $t$-spanner of $H$. 
Observe that this result holds for any weighted graph; in particular, 
it is valid for any geometric graph. For geometric graphs, a lower bound 
was given by Gudmundsson and Smid~\cite{gs-osogg-06}: They proved that 
for every real number $t$ with $1 < t < \frac{1}{4} \log n$, there exists 
a geometric graph $H$ with $n$ vertices, such that every $t$-spanner 
of $H$ contains $\Omega( n^{1 + 1/t} )$ edges.
Thus, if we are looking for spanners with $O(n)$ edges of arbitrary 
geometric graphs, then the best stretch factor we can obtain is 
$\Theta(\log n)$. 

In this paper, we consider the case when the input graph is a complete 
$k$-partite geometric graph. Let $S$ be a set of $n$ points in 
$\mathbb{R}^d$, and let $S$ be partitioned into subsets 
$C_1,C_2,\ldots,C_k$. Let $K_{C_1\ldots C_k}$ denote the 
\emph{complete $k$-partite graph on $S$}. This graph has $S$ as its 
vertex set and two points $p$ and $q$ are connected by an edge 
(of length $|pq|$) if and only if $p$ and $q$ are in different subsets 
of the partition. The problem we address is formally defined as 
follows:

\begin{problem} 
Let $k \geq 2$ be an integer, let $S$ be a set of $n$ points in 
$\mathbb{R}^d$, and let $S$ be partitioned into $k$ subsets 
$C_1,C_2,\ldots,C_k$. Compute a $t$-spanner of the complete $k$-partite 
graph $K_{C_1\ldots C_k}$ that has a ``small'' number of edges and whose 
stretch factor $t$ is ``small''. 
\end{problem}

The main contribution of this paper is to present an algorithm that 
computes such a $t$-spanner with $O(n)$ edges in $O(n \log n)$ time, 
where $t=5+\epsilon$ for any constant $\epsilon>0$. We also show that 
if one is willing to use $O(n \log n)$ edges, then our algorithm adapts 
easily to reach a stretch factor of $t=3+\epsilon$. Finally, we show 
that the latter result is optimal: For any $k$ with 
$2 \leq k \leq n - \Theta(\sqrt{n \log n} )$, spanners with 
$O(n \log n)$ edges and stretch factor less than $3$ do not exist 
for all complete $k$-partite geometric graphs. 

We remark that in a recent paper, 
Bose \etal~\cite{couture07chromaticReport} 
considered the problem of constructing spanners of point sets that have 
$O(n)$ edges and whose chromatic number is a most $k$. This problem 
is different from ours: Bose \etal\ compute a spanner of the 
complete graph and their algorithm can choose a ``good'' $k$-partition 
of the vertices. In our problem, the $k$-partition is given and we want 
to compute a spanner of the complete $k$-partite graph.  
 
Possible applications of our algorithm are in wireless networks
having the property that communicating nodes are partitioned into
sets such that two nodes can communicate if and only if they do not
belong to the same set. This would be the case, for example, when
Time Division Multiplexing (TDMA) is used. Since the wireless medium
prohibits simultaneous transmission and reception at one node, two
nodes communicating during the same time slots cannot communicate
with each other; see Raman and Chebrolu~\cite{raman05}.

The rest of this paper is organized as follows. 
In Section~\ref{section-wspd}, we recall properties of the
Well-Separated Pair Decomposition (WSPD) that we use in our
algorithm. In Section~\ref{section-first-algo}, we provide an 
algorithm that solves the problem of constructing a spanner of the 
complete $k$-partite graph. In Section~\ref{section-first-analysis}, we 
show that the spanner constructed by this algorithm has $O(n)$ edges 
and that its stretch factor is bounded from above by a constant that 
depends only on the dimension $d$. In 
Section~\ref{section-improved-algo}, we show how a simple modification 
to our algorithm improves the stretch factor to $5+\epsilon$ while still 
having $O(n)$ edges. In Section~\ref{section-tightBounds}, we show how 
to achieve a stretch factor of $3+\epsilon$ using $O(n\log n)$ edges. 
We also prove that the latter result is optimal.  
We conclude in Section~\ref{section-bi-span-conclusion}.

\section{The Well-Separated Pair Decomposition}\label{section-wspd}
In this section, we recall crucial properties of the 
Well-Separated Pair Decomposition (WSPD) of 
Callahan and Kosaraju~\cite{callahan95} that we use for our 
construction. The reader who is familiar with the WSPD may go 
directly to Section~\ref{section-first-algo}.
Our presentation follows the one in Narasimhan and
Smid~\cite{smid07}. Intuitively, a WSPD is a partition of the
edges of a complete geometric graph such that all edges that are
grouped together are \emph{approximately} equal. To give a formal
definition of the WSPD, we first need to define what it means for two
sets to be well-separated.

\begin{definition} 
Let $S$ be a set of points in $\mathbb{R}^d$. The \emph{bounding box}
$\beta(S)$ of $S$ is the smallest axes-parallel hyperrectangle that 
contains $S$.
\end{definition}

\begin{definition} 
Let $X$ and $Y$ be two sets of points in $\mathbb{R}^d$ and let 
$s>0$ be a real number. We say that $X$ and $Y$ are 
\emph{well-separated} with respect to $s$ if there exists two balls 
$B_1$ and $B_2$ such that 
\begin{enumerate}
\item $B_1$ and $B_2$ have the same radius, say $\rho$,
\item $B_1$ contains the bounding box of $X$,
\item $B_2$ contains the bounding box of $Y$, and
\item the distance 
      $\min \{ |xy| : x \in B_1 \cap \mathbb{R}^d , 
                      y \in B_2 \cap \mathbb{R}^d 
              \}$ 
      between $B_1$ and $B_2$ is at least $s \rho$.
\end{enumerate}
\end{definition}

\begin{definition} 
Let $S$ be a set of points in $\mathbb{R}^d$ and let $s>0$ be a real 
number. A \emph{well-separated pair decomposition (WSPD) of $S$ with 
separation constant $s$} is a set of unordered pairs of subsets of 
$S$ that are well-separated with respect to $s$, such that for any 
two distinct points $p,q\in S$ there is a unique pair $\{X,Y\}$ in the 
WSPD such that $p\in X$ and $q\in Y$.
\end{definition}

\begin{lemma}[Lemma 9.1.2 in \cite{smid07}]  \label{lemma-912} 
Let $s>0$ be a real number and let $X$ and $Y$ be two point sets 
that are well-separated with respect to $s$. 
\begin{enumerate}
\item If $p,p',p'' \in X$ and $q \in Y$, then 
      $|p'p''|\leq (2/s)|pq|$. 
\item If $p,p' \in X$ and $q,q' \in Y$, then 
      $|p'q'|\leq (1+4/s)|pq|$.
\end{enumerate}
\end{lemma}

The first part of this lemma states that distances within one set are 
very small compared to distances between pairs of points having one 
endpoint in each set. The second part states that all pairs of points 
having one endpoint in each set have approximately the same distance.

Callahan and Kosaraju~\cite{ck-fasgg-93} have shown how to construct a 
$t$-spanner of $S$ from a WSPD: All one has to do is pick from each 
pair $\{X,Y\}$ an arbitrary edge $(p,q)$ with $p\in X$ and $q\in Y$. 
Using induction on the rank of the length of the edges in the complete 
graph $K_S$, it can be shown that, when $s>4$, this process leads to 
a $((s+4)/(s-4))$-spanner. Thus, by choosing $s$ to be a sufficiently 
large constant, the stretch factor can be made arbitrarily close to $1$. 

In order to compute a spanner of $S$ that has a linear number of edges, 
one needs a WSPD that has a linear number of pairs. Callahan and
Kosaraju~\cite{callahan95} showed that a WSPD with a linear number
of pairs always exists and can be computed in time $O(n\log n)$.
Their algorithm uses a split-tree.

\begin{definition} 
Let $S$ be a non-empty set of points in $\mathbb{R}^d$. The 
\emph{split-tree} of $S$ is defined as follows: if $S$ contains only 
one point, then the split-tree is a single node that stores that point.
Otherwise, the split-tree has a root that stores the bounding box 
$\beta(S)$ of $S$, as well as an arbitrary point of $S$ called the 
\emph{representative} of $S$ and denoted by $\rep(S)$. Split 
$\beta(S)$ into two hyperrectangles by cutting its longest interval 
into two equal parts, and let $S_1$ and $S_2$ be the subsets of $S$ 
contained in the two hyperrectangles. The root of the split-tree 
of $S$ has two sub-trees, which are recursively defined split-trees 
of $S_1$ and $S_2$.
\end{definition}

The precise way Callahan and Kosaraju used the split-tree to compute a
WSPD with a linear number of pairs is of no importance to us. 
The only important aspect we need to retain is that each pair is 
uniquely determined by a pair of nodes in the tree. More precisely, for 
each pair $\{X,Y\}$ in the WSPD that is output by their algorithm, 
there are unique internal nodes $u$ and $v$ in the split-tree such that 
the sets $S_u$ and $S_v$ of points stored at the leaves of the subtrees 
rooted at $u$ and $v$ are precisely $X$ and $Y$. Since there is such 
a unique correspondence, we will denote pairs in the WSPD by 
$\{S_u,S_v\}$, meaning that $u$ and $v$ are the nodes corresponding to 
the sets $X = S_u$ and $Y = S_v$. Also, although the WSPD of a point 
set is not unique, when we talk about \emph{the} WSPD, we mean the 
WSPD that is computed by the algorithm of Callahan and Kosaraju.

Before we present our algorithm, we give the statement of the
following lemmas that we use to analyze our algorithm in
Section~\ref{section-first-analysis}.
If $R$ is an axes-parallel hyperrectangle in $\mathbb{R}^d$, then 
we use $L_{\max}(R)$ to denote the length of a longest side of $R$. 
 
\begin{lemma}[Lemma~9.5.3 in \cite{smid07}]  \label{lemma-953}
Let $u$ be a node in the split-tree and let $u'$ be a node in the 
subtree of $u$ such that the path between them contains at least 
$d$ edges. Then
\[ L_{\max}(\beta(S_{u'}))\leq \frac{1}{2}\cdot L_{\max}(\beta(S_u)).
\] 
\end{lemma}

\begin{lemma}[Lemma~11.3.1 in \cite{smid07}]  \label{lemma-1131}
Let $\{S_u,S_v\}$ be a pair in the WSPD, let $\ell$ be the distance 
between the centers of $\beta(S_u)$ and $\beta(S_v)$, and let $\pi(u)$ 
be the parent of $u$ in the split-tree. Then
\[ L_{\max}(\beta(S_{\pi(u)})) \geq \frac{2\ell}{\sqrt{d}(s+4)}.
\]
\end{lemma}

%

\section{A First Algorithm}\label{section-first-algo}

We now show how the WSPD can be used to address the problem of
computing a spanner of a complete $k$-partite graph. In this section,  
we introduce an algorithm that outputs a graph with constant
stretch factor and $O(n)$ edges. The analysis of this 
algorithm is presented in Section~\ref{section-first-analysis}.
In Section~\ref{section-improved-algo}, we show how this algorithm can
be improved to achieve a stretch factor of $5+\epsilon$.

The input set $S \subseteq \mathbb{R}^d$ is the disjoint union of $k$ 
sets $C_1,C_2,\ldots,C_k$. We say that the elements of $C_c$ have 
``color'' $c$. The graph $K = K_{C_1 \ldots C_k}$ is the complete 
$k$-partite geometric graph. 

\begin{definition}   \label{defMWSPD}   
Let $T$ be the split-tree of $S$ that is used to compute the WSPD of $S$.
\begin{enumerate}
\item For any node $u$ in $T$, we denote by $S_u$ the set of all points 
      in the subtree rooted at $u$.
\item We define MWSPD to be the subset of the WSPD obtained by removing 
      all pairs $\{S_u,S_v\}$ for which all points of $S_u \cup S_v$ 
      have the same color. 
\item A node $u$ in $T$ is called \emph{multichromatic} if there exist 
      points $p$ and $q$ in $S_u$ and a node $v$ in $T$, such that 
      $p$ and $q$ have different colors and $\{S_u,S_v\}$ is in the 
      MWSPD.
\item A node $u$ in $T$ is called a \emph{$c$-node} if all points of 
      $S_u$ have color $c$ and there exists a node $v$ in $T$ such that 
      $\{S_u,S_v\}$ is in the MWSPD. 
\item A $c$-node $u$ in $T$ is called a \emph{$c$-root} if it does not 
      have a proper ancestor that is a $c$-node in $T$.
\item A $c$-node $u$ in $T$ is called a \emph{$c$-leaf} if it does not 
      have another $c$-node in its subtree.
\item A $c$-node $u'$ in $T$ is called a \emph{$c$-child} of a $c$-node 
      $u$ in $T$ if $u'$ is in the subtree rooted at $u$ and there is 
      no $c$-node on the path strictly between $u$ and $u'$.
\item For each color $c$ and for each $c$-node $u$ in $T$, 
      $\rep(S_u)$ denotes a fixed arbitrary point in $S_u$.
\item For each multichromatic node $u$ in $T$, $\rep(S_u)$ and 
      $\rep'(S_u)$ denote two fixed arbitrary points in $S_u$ that 
      have different colors.
\item The \emph{distance} between two sets $S_v$ and $S_w$, denoted by 
      $\dist(S_v,S_w)$, is defined to be the distance between the 
      centers of their bounding boxes.
\item Let $u$ be a $c$-node in $T$. Consider all pairs $\{ S_v,S_w \}$ 
      in the MWSPD, where $v$ is a $c$-node on the path in $T$ from 
      $u$ to the root (this path includes $u$). Let $\{ S_v,S_w \}$ be 
      such a pair for which $\dist(S_v,S_w)$ is minimum. We define 
      $\cl(S_u)$ to be the set $S_w$. 
\end{enumerate}
\end{definition}

\begin{figure}
\centering
\includegraphics{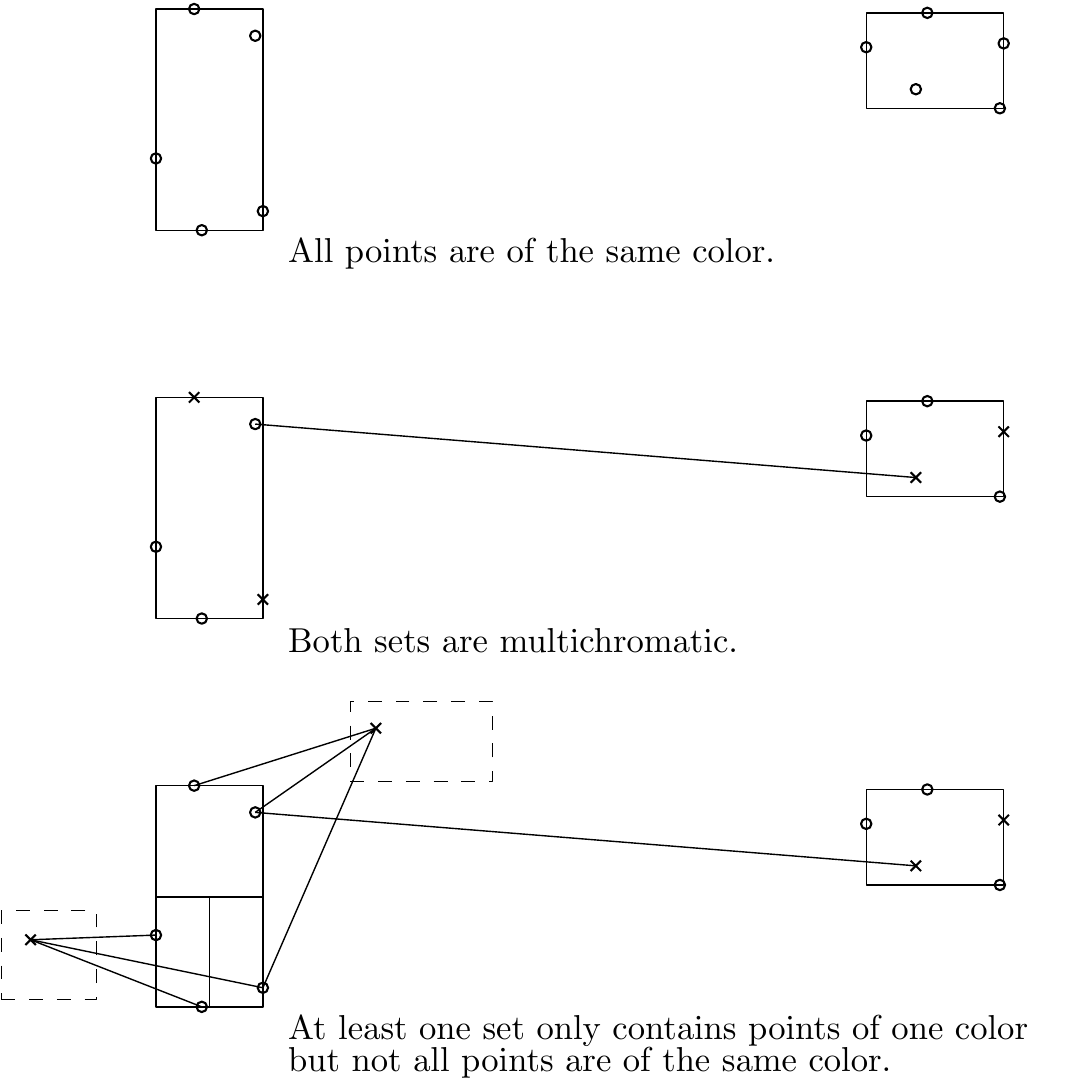}
\caption{The three cases of Algorithm~\ref{alg-biSpan}.}
\label{fig-bi-cases}
\end{figure} 

\begin{figure} 
\centering
\includegraphics{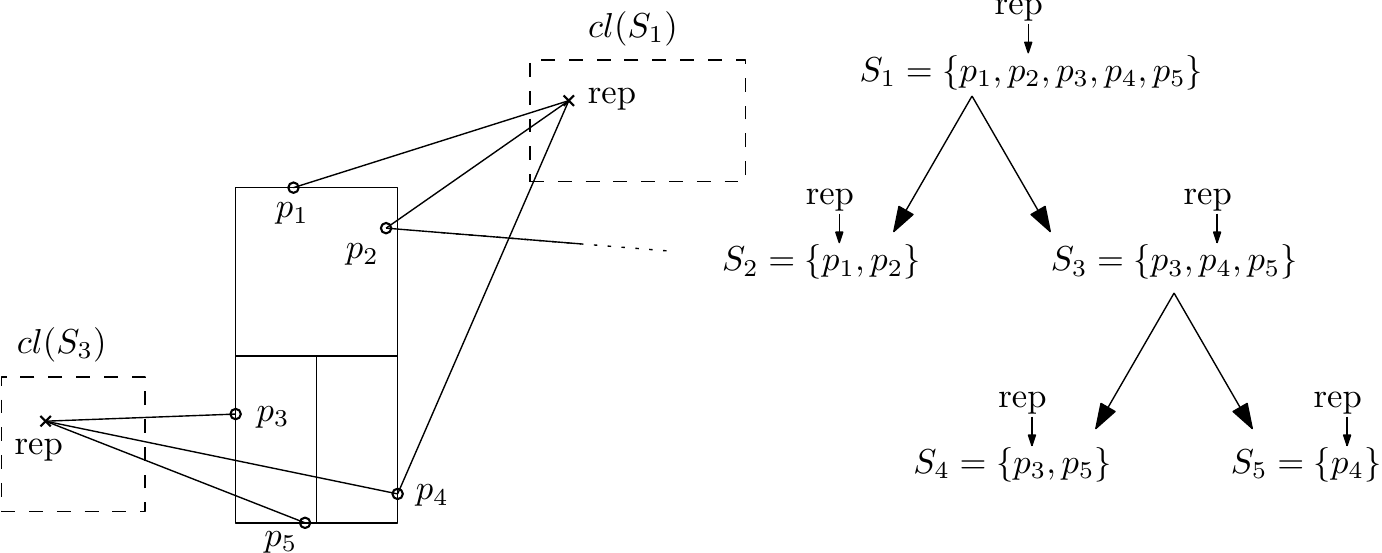}
\caption{Handling a $c$-node.}
\label{fig-go-to-rep}
\end{figure}

Algorithm~\ref{alg-biSpan} computes a spanner of a complete
$k$-partite geometric graph $K = K_{C_1 \ldots C_k}$. The intuition 
behind this algorithm is the following. First, the algorithm computes 
the WSPD. Then, it considers each pair $\{S_u,S_v\}$ of the WSPD, and 
decides whether or not to add an edge between $S_u$ and $S_v$. The 
outcome of this decision is based on the following three cases, which 
are illustrated in Figures~\ref{fig-bi-cases} and~\ref{fig-go-to-rep}.

\vspace{0.5em} 

\noindent 
{\bf Case 1:} 
All points of $S_u \cup S_v$ are of the same color. In this case, there 
is no edge of $K$ to approximate, so the algorithm ignores this 
pair. 

\vspace{0.5em} 

\noindent 
{\bf Case 2:} 
Both $S_u$ and $S_v$ are multichromatic. In this case, the algorithm 
adds one edge between $S_u$ and $S_v$ to the spanner; see lines 
\ref{line28}--\ref{line29}. Observe that the two vertices of this edge 
do not have the same color. This edge will allow us to approximate 
each edge $(p,q)$ of $K$, where $p \in S_u$, $q \in S_v$, and $p$ 
and $q$ have different colors. 

\vspace{0.5em} 

\noindent 
{\bf Case 3:} 
All points in $S_u$ are of the same color $c$. In this case, an edge 
is added between $\rep(S_u)$ and one of the two representatives of 
$S_v$ whose color is not $c$; see lines \ref{line17}--\ref{line18}.
In order to approximate each edge of $K$ having one vertex (of color 
$c$) in $S_u$ and the other vertex (of a different color) in $S_v$, 
more edges have to be added. This is done in such a way that our final 
graph contains a ``short'' path between every point $p$ of $S_u$ and the
representative $\rep(S_u)$ of $S_u$. Observe that this path must contain 
points whose color is not equal to $c$; thus, these points are not 
in $S_u$. One way to achieve this is to add an edge between each point 
of $S_u$ and one of the two representatives of $\cl(S_u)$ whose color 
is not $c$; we call this construction a \emph{star}.
However, since the subtree rooted at $u$ may contain other $c$-nodes,
many edges may be added for each point in $S_u$, which could possibly 
lead to a quadratic number of edges in the final graph. To guarantee 
that the algorithm does not add too many edges, it introduces a star 
only if $u$ is a $c$-leaf; see lines \ref{line8}--\ref{line11}.  
If $u$ is a $c$-node, the algorithm only adds one edge between  
$\rep(S_u)$ and a representatives of $\cl(S_u)$ whose color is 
not $c$; see lines \ref{line14}--\ref{line15}. Then, the algorithm links 
each $c$-node $u''$ that is not a $c$-root to its $c$-parent $u'$. This 
is done through an edge between $\rep(S_{u''})$ and a representative 
of $\cl(S_{u'})$ whose color is not $c$; see lines 
\ref{line21}--\ref{line22}. This third case is illustrated in 
Figure~\ref{fig-go-to-rep}.

\begin{algorithm} 
\SetLine
\KwIn{A set $S$ of points in $\mathbb{R}^d$, which is partitioned 
into $k$ subsets $C_1,\ldots,C_k$.}  
\KwOut{A spanner $G=(S,E)$ of the complete $k$-partite graph 
$K_{C_1\ldots C_k}$.} 
\BlankLine  
compute the split-tree $T$ of $S$\; 
using $T$, compute the WSPD with respect to a separation 
       constant $s>0$\;  
using the WSPD, compute the MWSPD\nllabel{line3}\; 
$E \leftarrow \emptyset$\; 
\For{\emph{each color $c$ in $\{ 1,2,\ldots, k \}$}}{ 
     \For{\emph{each $c$-root $u$ in $T$}}{ 
          \For{\emph{each $c$-leaf $u'$ in the subtree of $u$}}{ 
               \For{\emph{each $p \in S_{u'}$}\nllabel{line8}}{  
                    \lIf{\emph{$\rep(\cl(S_{u'}))$ does not have 
                          color $c$}}
                         {add $(p,\rep(\cl(S_{u'})))$ 
                           to $E$}\nllabel{line9}\; 
                    \lElse{add $(p,\rep'(\cl(S_{u'})))$ 
                           to $E$}\nllabel{line10}\; 
                    }\nllabel{line11} 
          }  
          \For{\emph{each $c$-node $u'$ that is in the subtree of $u$ 
               (including $u$)}}{  
               \lIf{\emph{$\rep(\cl(S_{u'}))$ does not have color $c$}}
                   {add $(\rep(S_{u'}),\rep(\cl(S_{u'})))$ 
                           to $E$}\nllabel{line14}\; 
                   \lElse{add $(\rep(S_{u'}),\rep'(\cl(S_{u'})))$ 
                           to $E$}\nllabel{line15}\; 
               \For{\emph{for each pair $\{S_{u'},S_{v'}\}$ in the 
                             MWSPD}}{
                    \lIf{\emph{$\rep(S_{v'})$ does not have color $c$}}
                        {add $(\rep(S_{u'}),\rep(S_{v'}))$ 
                                to $E$}\nllabel{line17}\;
                    \lElse{add $(\rep(S_{u'}),\rep'(S_{v'}))$ 
                                to $E$}\nllabel{line18}\; 
                    } 
               \For{\emph{each $c$-child $u''$ of 
                    $u'$}\nllabel{line20}}{ 
                    \lIf{\emph{$\rep(\cl(S_{u'}))$ does not have 
                                color $c$}}
                        {add $(\rep(S_{u''}),\rep(\cl(S_{u'})))$ 
                                      to $E$}\nllabel{line21}\; 
                    \lElse{add $(\rep(S_{u''}),\rep'(\cl(S_{u'})))$ 
                                      to $E$}\nllabel{line22}\; 
                    }\nllabel{line23} 
           } 
     }  
} 
\For{\emph{each $\{S_{u},S_{v}\}$ in the MWSPD for which both $u$ and 
    $v$ are multichromatic}}{
     \lIf{\emph{$\rep(S_u)$ and $\rep(S_v)$ have distinct colors}}
          {add $(\rep(S_u),\rep(S_v))$ to $E$}\nllabel{line28}\; 
     \lElse{add $(\rep(S_u),\rep'(S_v))$ 
                    to $E$}\nllabel{line29}\; 
    } 
return the graph $G=(S,E)$\nllabel{line31} 
\caption{Computing a sparse subgraph of $K_{C_1 \ldots C_k}$ whose 
    stretch factor is bounded by a constant.\label{alg-biSpan}} 
\end{algorithm}

\section{Analysis of Algorithm~\ref{alg-biSpan}}
\label{section-first-analysis}

\begin{lemma}
\label{thm-numberOfEdges}
The graph $G$ computed by Algorithm~\ref{alg-biSpan} has $O(|S|)$ edges.
\end{lemma}
\begin{proof}
For each color $c$ and for each $c$-leaf $u'$, the algorithm adds 
$|S_{u'}|$ edges to $G$ in lines \ref{line9}--\ref{line10}. 
Since the sets $S_{u'}$, where $u'$ ranges over all $c$-leaves and $c$ 
ranges over all colors, are pairwise disjoint, the total number of 
edges that are added in lines \ref{line9}--\ref{line10} is $O(|S|)$. 

The total number of edges that are added in lines 
\ref{line17}--\ref{line18} and \ref{line28}--\ref{line29} is at most 
the number of pairs in the MWSPD. Since the WSPD contains $O(|S|)$ 
pairs (see~\cite{callahan95}), the same upper bound holds for the 
number of edges added in lines \ref{line17}--\ref{line18} and 
\ref{line28}--\ref{line29}.  

The total number of edges that are added in lines 
\ref{line14}--\ref{line15} and \ref{line21}--\ref{line22}  
is at most twice the number of nodes in the split-tree, which is 
$O(|S|)$. 
\end{proof}

\begin{lemma}\label{lemma-levels}
Let $G$ be the graph computed by Algorithm~\ref{alg-biSpan}. 
Let $p$ and $q$ be two points of $S$ with different colors, and let 
$\{S_u, S_v\}$ be the pair in the MWSPD for which $p \in S_u$ and 
$q \in S_v$. Assume that $u$ is a $c$-node for some color $c$. Then 
there is a path in $G$ between $p$ and $\rep(S_u)$ whose length is 
at most $t' |pq|$, where 
\[ t' = 4 \sqrt{d} (\mu d+1) (1+4/s)^3 ,   
\] 
\[ \mu = \left\lceil \log \left( \sqrt{d} (1+4/s) \right) 
         \right\rceil + 1 , 
\]
and $s$ is the separation constant of the WSPD.  
\end{lemma}
\begin{proof} 
Let $w$ be the $c$-leaf such that $p \in S_w$, and let 
$w=w_0,\ldots,w_k=u$ be the sequence of $c$-nodes that are on the path 
in $T$ from $w$ to $u$. 

Recall from Definition~\ref{defMWSPD} that each set $S_{w_i}$, 
$0 \leq i \leq k$, has a representative $\rep(S_{w_i})$ (of color $c$) 
associated with it. Also, recall the definition of the sets 
$\cl(S_{w_i})$, $0 \leq i \leq k$; see Definition~\ref{defMWSPD}. 
If $\cl(S_{w_i})$ is a $c'$-node for some color $c'$, then this set 
has one representative $\rep(\cl(S_{w_i}))$ associated with it. 
Otherwise, $\cl(S_{w_i})$ is multichromatic and this set has two 
representatives $\rep(\cl(S_{w_i}))$ and $\rep'(\cl(S_{w_i}))$ of 
different colors associated with it. We may assume without loss of 
generality that, for all $0 \leq i \leq k$, the color of 
$\rep(\cl(S_{w_i}))$ is not equal to $c$.  
 
Let $\Pi$ be the path  
\[ 
\begin{array}{ccccccc} 
  p & \rightarrow & \rep(\cl(S_{w_0})) & \rightarrow & \rep(S_{w_0}) 
         & & \\
  & \rightarrow & \rep(\cl(S_{w_1})) & \rightarrow & \rep(S_{w_1}) 
         & & \\ 
  & \vdots & & \vdots & & & \\ 
  & \rightarrow & \rep(\cl(S_{w_k})) & \rightarrow & \rep(S_{w_k}) 
              & = & \rep(S_u) . 
\end{array} 
\]  
The first edge on this path, i.e., $(p , \rep(\cl(S_{w_0})))$, is 
added to the graph $G$ in lines \ref{line9}--\ref{line10} of the 
algorithm. The edges $( \rep(\cl(S_{w_i})) , \rep(S_{w_i}) )$, 
$0 \leq i \leq k$, are added to $G$ in lines 
\ref{line14}--\ref{line15}. Finally, the edges  
$( \rep(S_{w_{i-1}}) , \rep(\cl(S_{w_i})) )$, $1 \leq i \leq k$, 
are added to $G$ in lines \ref{line21}--\ref{line22}.
It follows that $\Pi$ is a path in $G$ between $p$ and $\rep(S_u)$. 
We will show that the length of $\Pi$ is at most $t' |pq|$. 

Let $i$ be an integer with $0 \leq i \leq k$. Recall the definition 
of $\cl(S_{w_i})$; see Definition~\ref{defMWSPD}: We consider all 
pairs $\{ S_x,S_y \}$ in the MWSPD, where $x$ is a $c$-node on the 
path in $T$ from $w_i$ to the root, and pick the pair for which 
$\dist(S_x,S_y)$ is minimum. We denote the pair picked by 
$(S_{x_i},S_{y_i})$. Thus, $x_i$ is a $c$-node on the path in $T$ 
from $w_i$ to the root, $\{ S_{x_i},S_{y_i} \}$ is a pair in the MWSPD, 
and $\cl(S_{w_i}) = S_{y_i}$. We define 
\[ \ell_i = \dist(S_{x_i},S_{y_i}) . 
\] 

Consider the first edge $(p , \rep(\cl(S_{w_0})))$ on the path 
$\Pi$. Since $p \in S_{w_0} \subseteq S_{x_0}$ and 
$\rep(\cl(S_{w_0})) \in S_{y_0}$, it follows from 
Lemma~\ref{lemma-912} that 
\[ | p , \rep(\cl(S_{w_0})) | \leq 
         (1+4/s) \cdot \dist(S_{x_0},S_{y_0}) = (1+4/s) \ell_0 . 
\] 
Let $0 \leq i \leq k$ and consider the edge 
$(\rep(\cl(S_{w_i})) , \rep(S_{w_i}))$ on $\Pi$. 
Since $\rep(S_{w_i}) \in S_{w_i} \subseteq S_{x_i}$ and 
$\rep(\cl(S_{w_i})) \in S_{y_i}$, it follows from 
Lemma~\ref{lemma-912} that
\begin{equation}     
   | \rep(\cl(S_{w_i})) , \rep(S_{w_i}) | \leq 
         (1+4/s) \cdot \dist(S_{x_i},S_{y_i}) = (1+4/s) \ell_i .  
      \label{paz23} 
\end{equation} 
Let $1 \leq i \leq k$ and consider the edge 
$(\rep(S_{w_{i-1}}) , \rep(\cl(S_{w_i})))$ on $\Pi$. 
Since $\rep(S_{w_{i-1}}) \in S_{w_{i-1}} \subseteq S_{x_i}$ 
and $\rep(\cl(S_{w_i})) \in S_{y_i}$, it follows from 
Lemma~\ref{lemma-912} that 
\[ | \rep(S_{w_{i-1}}) , \rep(\cl(S_{w_i})) | \leq  
         (1+4/s) \cdot \dist(S_{x_i},S_{y_i}) = (1+4/s) \ell_i . 
\]
Thus, the length of the path $\Pi$ is at most 
\[ \sum_{i=0}^k 2 (1+4/s) \ell_i . 
\] 
Therefore, it is sufficient to prove that 
\[ \sum_{i=0}^k \ell_i \leq 2 \sqrt{d} (\mu d+1) ( 1+4/s )^2 |pq| .    
\] 
It follows from the definition of $\cl(S_u) = \cl(S_{w_k})$ that 
\[ \ell_k = \dist(S_{x_k},S_{y_k}) \leq \dist(S_u,S_v) . 
\] 
Since, by Lemma~\ref{lemma-912}, $\dist(S_u,S_v) \leq (1+4/s) |pq|$, 
it follows that 
\begin{equation}    \label{paz10}  
       \ell_k \leq (1+4/s) |pq| .
\end{equation} 
Thus, it is sufficient to prove that 
\begin{equation}    \label{eqtoshow} 
    \sum_{i=0}^k \ell_i \leq 2 \sqrt{d} (\mu d+1) (1+4/s) \ell_k .    
\end{equation} 

If $k=0$, then (\ref{eqtoshow}) obviously holds. Assume from now 
on that $k \geq 1$. For each $i$ with $0 \leq i \leq k$, we define 
\[ a_i = L_{\max}(\beta(S_{w_i})) , 
\] 
i.e., $a_i$ is the length of a longest side of the bounding box of 
$S_{w_i}$. 

Let $0 \leq i \leq k$. It follows from Lemma~\ref{lemma-912} that 
\[ L_{\max}(\beta(S_{x_i})) \leq \frac{2}{s} \ell_i . 
\] 
Since $w_i$ is in the subtree of $x_i$, we have 
$L_{\max}(\beta(S_{w_i})) \leq L_{\max}(\beta(S_{x_i}))$.  
Thus, we have 
\begin{equation} 
  a_i \leq \frac{2}{s} \ell_i \mbox{ for $0 \leq i \leq k$.}   
             \label{paz1}  
\end{equation} 
Lemma~\ref{lemma-953} states that 
\begin{equation} 
  a_i \leq \frac{1}{2} a_{i+d}  \mbox{ for $0 \leq i \leq k-d$.} 
             \label{paz2}  
\end{equation} 
Let $0 \leq i \leq k-1$. Since $w_i$ is a $c$-node, there is a 
node $w'_i$ such that $\{ S_{w_i} , S_{w'_i} \}$ is a pair in 
the MWSPD. Then it follows from the definition of $\cl(S_{w_i})$ that 
\[ \ell_i = \dist(S_{x_i},S_{y_i}) \leq \dist(S_{w_i},S_{w'_i}) . 
\] 
By applying Lemma~\ref{lemma-1131}, we obtain 
\begin{eqnarray*}   
  \dist(S_{w_i},S_{w'_i}) & \leq &  
    \frac{\sqrt{d}(s+4)}{2} \, L_{\max}(\beta(S_{\pi(w_i)})) \\ 
  & \leq & \frac{\sqrt{d}(s+4)}{2} \, L_{\max}(\beta(S_{w_{i+1}})) \\ 
  & = & \frac{\sqrt{d}(s+4)}{2} \, a_{i+1} . 
\end{eqnarray*}   
Thus, we have 
\begin{equation} 
  \ell_i \leq \frac{\sqrt{d}(s+4)}{2} \, a_{i+1} 
          \mbox{ for $0 \leq i \leq k-1$.} 
           \label{paz3} 
\end{equation} 

First assume that $1 \leq k \leq \mu d$. Let $0 \leq i \leq k-1$. 
By using (\ref{paz3}), the fact that the sequence $a_0,a_1,\ldots,a_k$ 
is non-decreasing, and (\ref{paz1}), we obtain  
\[ \ell_i \leq \frac{\sqrt{d}(s+4)}{2} \, a_{i+1} \leq 
     \frac{\sqrt{d}(s+4)}{2} \, a_k \leq 
     \sqrt{d} (1+4/s) \ell_k .
\] 
Therefore, 
\[ \sum_{i=0}^k \ell_i \leq k \sqrt{d} (1+4/s) \ell_k + \ell_k  
      \leq (k+1) \sqrt{d} (1+4/s) \ell_k 
      \leq (\mu d +1) \sqrt{d} (1+4/s) \ell_k , 
\] 
which is less than the right-hand side in (\ref{eqtoshow}). 

It remains to consider the case when $k > \mu d$. 
Let $i \geq 0$ and $j \geq 0$ be integers such that $i+1+jd \leq k$. 
By applying (\ref{paz3}) once, (\ref{paz2}) $j$ times, 
and (\ref{paz1}) once, we obtain 
\[ \ell_i \leq \frac{\sqrt{d} (s+4)}{2} \, a_{i+1}  
          \leq \frac{\sqrt{d} (s+4)}{2} \left( \frac{1}{2} \right)^j 
                     a_{i+1+jd}   
          \leq \sqrt{d} (1+4/s) \left( \frac{1}{2} \right)^j 
                    \ell_{i+1+jd} . 
\]
For 
$j = \mu = \lceil \log ( \sqrt{d} (1+4/s) ) \rceil + 1$, 
this implies that, for $0 \leq i \leq k - 1 - \mu d$,  
\begin{equation}    \label{paz4} 
    \ell_i \leq \frac{1}{2} \ell_{i+1 + \mu d} .
\end{equation} 
By re-arranging the terms in the summation in (\ref{eqtoshow}), we 
obtain 
\[ \sum_{i=0}^k \ell_i = 
      \sum_{h=0}^{\mu d}  \ \ 
      \sum_{j=0}^{\lfloor (k-h)/(\mu d +1) \rfloor} 
              \ell_{k - h - j (\mu d + 1)} . 
\]  
Let $j$ be such that $0 \leq j \leq \lfloor (k-h)/(\mu d +1) \rfloor$. 
By applying (\ref{paz4}) $j$ times, we obtain 
\[ \ell_{k - h - j (\mu d + 1)} \leq 
    \left( \frac{1}{2} \right)^j \ell_{k-h} .
\] 
It follows that 
\[ \sum_{j=0}^{\lfloor (k-h)/(\mu d +1) \rfloor} 
          \ell_{k - h - j (\mu d + 1)}  
      \leq 
   \sum_{j=0}^{\infty} \left( \frac{1}{2} \right)^j \ell_{k-h}  
         = 2 \ell_{k-h} . 
\] 
Thus, we have 
\[ \sum_{i=0}^k \ell_i \leq 2 \sum_{h=0}^{\mu d} \ell_{k-h} . 
\] 
By applying (\ref{paz3}), the fact that the sequence 
$a_0,a_1,\ldots,a_k$ is non-decreasing, followed by (\ref{paz1}), 
we obtain, for $0 \leq i \leq k-1$ and $1 \leq j \leq k-i$,  
\[ \ell_i \leq \frac{\sqrt{d}(s+4)}{2} \, a_{i+1} 
         \leq \frac{\sqrt{d}(s+4)}{2} \, a_{i+j}  
         \leq \sqrt{d} (1+4/s) \ell_{i+j} .  
\]  
Obviously, the inequality $\ell_i \leq \sqrt{d} (1+4/s) \ell_{i+j}$ 
also holds for $j=0$. Thus, for $i=k-h$ and $j=h$, we get 
\[ \ell_{k-h} \leq \sqrt{d} (1+4/s) \ell_k 
          \mbox{ for $0 \leq h \leq \mu d$.} 
\] 
It follows that
\[ \sum_{i=0}^k \ell_i \leq  
     2 \sum_{h=0}^{\mu d} \sqrt{d} (1+4/s) \ell_k  
     = 2 \sqrt{d} (\mu d  + 1) (1+4/s) \ell_k , 
\] 
completing the proof that (\ref{eqtoshow}) holds.  
\end{proof}

\begin{lemma}
\label{thm-pathS}
Assuming that the separation constant $s$ of the WSPD is chosen 
sufficiently large, the graph $G$ computed by 
Algorithm~\ref{alg-biSpan} is a $t$-spanner of the complete 
$k$-partite graph $K_{C_1 \ldots C_k}$, where $t=2t'+1+4/s$ and 
$t'$ is as in Lemma~\ref{lemma-levels}. 
\end{lemma}
\begin{proof}
We denote the graph $K_{C_1 \ldots C_k}$ by $K$. 
It suffices to show that for each edge $(p,q)$ of $K$, the graph 
$G$ contains a path between $p$ and $q$ of length at most $t|pq|$. 
We will prove this by induction on the lengths of the edges in $K$. 

Let $p$ and $q$ be two points of $S$ with different colors, and let 
$\{S_u, S_v\}$ be the pair in the MWSPD for which $p \in S_u$ and 
$q \in S_v$. 

The base case is when $(p,q)$ is a shortest edge in $K$.  
Since $s>2$, it follows from Lemma~\ref{lemma-912} that $u$ is a 
$c$-node and $v$ is a $c'$-node, for some colors $c$ and $c'$ with 
$c \neq c'$. In line~\ref{line17} of Algorithm~\ref{alg-biSpan}, the 
edge $(\rep(S_u),\rep(S_v))$ is added to $G$. By Lemma~\ref{lemma-912}, 
the length of this edge is at most $(1+4/s)|pq|$. The claim follows 
from two applications of Lemma~\ref{lemma-levels} to get from $p$ to 
$\rep(S_u)$ and from $\rep(S_v)$ to $q$.

In the induction step, we distinguish four cases. 

\vspace{0.5em} 

\noindent 
{\bf Case 1:} $u$ is a $c$-node and $v$ is a $c'$-node, for some 
colors $c$ and $c'$ with $c \neq c'$. 

This case is identical to the base case.

\vspace{0.5em} 

\noindent 
{\bf Case 2:} $u$ is a $c$-node for some color $c$ and $v$ is a 
multichromatic node. 

In lines \ref{line17}--\ref{line18}, the edge 
$( \rep(S_u) , \rep(S_v) )$ or $( \rep(S_u) , \rep(S'_v) )$ 
is added to $G$. We may assume without loss of generality that 
$( \rep(S_u) , \rep(S_v) )$ is added. 
By Lemma~\ref{lemma-912}, the length of this edge is at most 
$(1+4/s)|pq|$. 

By Lemma~\ref{lemma-levels}, there is a path in $G$ between $p$ and 
$\rep(S_u)$ whose length is at most $t'|pq|$. 

First assume that $q$ and $\rep(S_v)$ have the same color. 
Let $r$ be a point in $S_v$ that has a color different from $q$'s 
color. Since $s>2$, it follows from Lemma~\ref{lemma-912} that 
$|qr|<|pq|$. Thus, by induction, there is a path in $G$ between $q$ 
and $r$ whose length is at most $t|qr|$, which, by 
Lemma~\ref{lemma-912}, is at most $(2t/s)|pq|$. 
By a similar argument, since $|r , \rep(S_v) | < |pq|$, 
there is a path in $G$ between $r$ and $\rep(S_v)$ whose length is 
at most $(2t/s)|pq|$. Thus, $G$ contains a path between $q$ and 
$\rep(S_v)$ of length at most $(4t/s)|pq|$. 
If $q$ and $\rep(S_v)$ have different colors, then, by induction, 
there is a path in $G$ between $q$ and $\rep(S_v)$ whose length is 
at most $(2t/s)|pq| < (4t/s)|pq|$. 

Thus, the graph $G$ contains a path between $q$ and $\rep(S_v)$ of 
length at most $(4t/s)|pq|$. 

We have shown that there is a path in $G$ between $p$ and $q$ 
whose length is at most 
\begin{equation}   
       \left( t' + (1+4/s) + 4t/s \right) |pq| .  \label{paz6} 
\end{equation}   
By choosing $s$ sufficiently large, this quantity is at most $t|pq|$. 

\vspace{0.5em} 

\noindent 
{\bf Case 3:} $u$ is a multichromatic node and $v$ is a $c$-node for 
some color $c$. 

This case is symmetric to Case 2. 

\vspace{0.5em} 

\noindent 
{\bf Case 4:} Both $u$ and $v$ are multichromatic nodes.

In lines \ref{line28}--\ref{line29}, the edge 
$( \rep(S_u) , \rep(S_v) )$ or $( \rep(S_u) , \rep(S'_v) )$ 
is added to $G$. We may assume without loss of generality that 
$( \rep(S_u) , \rep(S_v) )$ is added. 
By Lemma~\ref{lemma-912}, the length of this edge is at most 
$(1+4/s)|pq|$. 

As in Case~2, the graph $G$ contains a path between $p$ and 
$\rep(S_u)$ of length at most $(4t/s)|pq|$, and a path between 
$q$ and $\rep(S_v)$ of length at most $(4t/s)|pq|$. 

It follows that there is a path in $G$ between $p$ and $q$ whose 
length is at most 
\begin{equation}  
       \left( (1+4/s) + 8t/s \right) |pq| .  \label{paz666}
\end{equation} 
By choosing $s$ sufficiently large, this quantity is at most $t|pq|$. 
\end{proof}

\begin{lemma}    \label{lemtime} 
The running time of Algorithm~\ref{alg-biSpan} is $O(n \log n)$, 
where $n = |S|$. 
\end{lemma} 
\begin{proof} 
Using the results of Callahan and Kosaraju~\cite{callahan95}, the 
split-tree $T$ and the WSPD can be computed in $O(n \log n)$ time. 
The representatives of all sets $S_u$ and all sets $\cl(S_u)$ can be 
computed in $O(n)$ time by traversing the split-tree in post-order 
and pre-order, respectively. The time for the rest of the algorithm, 
i.e., lines \ref{line3}--\ref{line31}, is proportional to the sum of 
the size of $T$, the number of pairs in the WSPD and the number of 
edges in the graph $G$. Thus, the rest of the algorithm takes $O(n)$ 
time. 
\end{proof} 

To summarize, we have shown the following: For any complete $k$-partite 
geometric graph $K$ whose vertex set has size $n$, 
Algorithm~\ref{alg-biSpan} computes a $t$-spanner of $K$ having $O(n)$ 
edges, where $t$ is given in Lemma~\ref{thm-pathS}. The running time 
of this algorithm is $O(n \log n)$. By choosing the separation constant 
$s$ sufficiently large, the stretch factor $t$ converges to  
\[ 8 \sqrt{d} 
   \left( d \left\lceil \frac{1}{2} \log d \right\rceil + d + 1 
   \right) + 1 .
\] 

In the next section, we show how to modify the algorithm so that 
the bound in Lemma~\ref{lemma-levels} is reduced, thus improving the 
stretch factor. The price to pay is in the number of edges in $G$, 
however, it is still $O(n)$.

\section{An Improved Algorithm}\label{section-improved-algo}
As before, we are given a set $S$ of $n$ points in $\mathbb{R}^d$ 
which is partitioned into $k$ subsets $C_1,C_2,\ldots,C_k$. 
Intuitively, the way to improve the bound of Lemma~\ref{lemma-levels} 
is by adding shortcuts along the path from each $c$-leaf to the 
$c$-root above it. More precisely, from (\ref{paz4}) in the proof of 
Lemma~\ref{lemma-levels}, we know that if we go $1+\mu d$ levels up in 
the split-tree $T$, then the length of the edge along the path doubles. 
Thus, for each $c$-node in $T$, we will add edges to all 
$2\delta (1+\mu d)$ $c$-nodes above it in $T$. Here, $\delta$ is an 
integer constant that is chosen such that the best result is obtained 
in the improved bound. 

\begin{definition} 
Let $c \in \{ 1,2,\ldots,k \}$, and let $u$ and $u'$ be $c$-nodes in 
the split-tree $T$ such that $u'$ is in the subtree rooted at $u$. 
For any integer $\zeta \geq 1$, we say that $u$ is $\zeta$ 
\emph{levels above} $u'$, if there are exactly $\zeta - 1$ $c$-nodes 
on the path strictly between $u$ and $u'$. We say that $u'$ is a 
\emph{$\zeta$-level $c$-child} of $u$ if $u$ is at most $\zeta$ levels 
above $u'$. 
\end{definition}

The improved algorithm is given as Algorithm~\ref{alg-improveB}. 
The following lemma generalizes Lemma~\ref{lemma-levels}. 

\begin{algorithm}
\linesnumberedhidden
\KwIn{A set $S$ of points in $\mathbb{R}^d$, which is partitioned 
      into $k$ subsets $C_1,\ldots,C_k$, and a real constant 
      $0 < \epsilon < 1$.}  
\KwOut{A $(5+\epsilon)$-spanner $G=(S,E)$ of the complete $k$-partite 
       graph $K_{C_1\ldots C_k}$.} 
\BlankLine   
Choose a separation constant $s$ such that $s \geq 12/\epsilon$ 
and $(1+4/s)^2 \leq 1 + \epsilon/36$ and choose an integer constant 
$\delta$ such that 
$\frac{2^{\delta}}{2^{\delta}-1} \leq 1 + \epsilon/36$.\\ 
The rest of the algorithm is the same as Algorithm~\ref{alg-biSpan}, 
except for lines \ref{line20}--\ref{line23}, which are replaced by the 
following:\\ 
\BlankLine 
$\zeta \leftarrow 2 \delta (\mu d +1)$\;    
\For{\emph{each $\zeta$-level $c$-child $u''$ of $u'$}}{ 
    \lIf{\emph{$\rep(\cl(S_{u'}))$ does not have color $c$}} 
        {add $(\rep(S_{u''}),\rep(\cl(S_{u'})))$ to $E$}\; 
    \lElse{add $(\rep(S_{u''}),\rep'(\cl(S_{u'})))$ to $E$}\; 
    \lIf{\emph{$\rep(\cl(S_{u''}))$ does not have color $c$}} 
        {add $(\rep(\cl(S_{u''})),\rep(S_{u'}))$ to $E$}\; 
    \lElse{add $(\rep'(\cl(S_{u''})),\rep(S_{u'}))$ to $E$}\; 
} 
\caption{Computing a sparse $(5+\epsilon)$-spanner of 
       $K_{C_1 \ldots C_k}$.\label{alg-improveB}}
\end{algorithm}

\begin{lemma}    \label{lem-improvB}
Let $G$ be the graph computed by Algorithm~\ref{alg-improveB}. 
Let $p$ and $q$ be two points of $S$ with different colors, and let 
$\{S_u, S_v\}$ be the pair in the MWSPD for which $p \in S_u$ and 
$q \in S_v$. Assume that $u$ is a $c$-node for some color $c$. Then 
there is a path in $G$ between $p$ and $\rep(S_u)$ whose length is 
at most $(2 + \epsilon/3) |pq|$. 
\end{lemma}
\begin{proof} 
Let $w$ be the $c$-leaf such that $r \in S_w$, and let 
$w=w_0,w_1,\ldots,w_k=u$ be the sequence of $c$-nodes that are on the 
path in $T$ from $w$ to $u$. As in the proof of Lemma~\ref{lemma-levels}, 
we assume without loss of generality that, for all $0 \leq i \leq k$, 
the color of $\rep(\cl(S_{w_i}))$ is not equal to $c$.  

Throughout the proof, we will use the variables $x_i$, $y_i$, $\ell_i$, 
and $a_i$, for $0 \leq i \leq k$, that were introduced in the proof of 
Lemma~\ref{lemma-levels}. 

We first assume that $0 \leq k \leq 2 \delta (\mu d +1)$. Let $\Pi$ 
be the path 
\[ p \rightarrow \rep(\cl(S_w)) \rightarrow \rep(S_u) . 
\] 
It follows from Algorithm~\ref{alg-improveB} that $\Pi$ is a path in 
$G$. Since $p \in S_w = S_{w_0} \subseteq S_{x_0}$ and 
$\rep(\cl(S_w)) = \rep(\cl(S_{w_0})) \in S_{y_0}$, it follows 
from Lemma~\ref{lemma-912} that 
\begin{equation} 
   | p , \rep(\cl(S_w)) | \leq (1+4/s) \cdot \dist(S_{x_0},S_{y_0}) 
          = (1+4/s) \ell_0 . 
     \label{paz11} 
\end{equation} 
Since $\{S_u,S_v\}$ is one of the pairs that is considered in the 
definition of $\cl(S_{w_0})$, we have 
$\dist(S_{x_0},S_{y_0}) \leq \dist(S_u,S_v)$. Again by 
Lemma~\ref{lemma-912}, we have $\dist(S_u,S_v) \leq (1+4/s)|pq|$. 
Thus, we have shown that 
\[ | p , \rep(\cl(S_w)) | \leq (1+4/s)^2 |pq| .
\] 
By the triangle inequality, we have 
\[ | \rep(\cl(S_w)) , \rep(S_u) | \leq 
   | \rep(\cl(S_w)) , p | + | p , \rep(S_u) | .
\] 
Since $p$ and $\rep(S_u)$ are both contained in $S_u$, it follows 
from Lemma~\ref{lemma-912} that $| p , \rep(S_u) | \leq (2/s)|pq|$. 
Thus, we have 
\[ | \rep(\cl(S_w)) , \rep(S_u) | \leq (1+4/s)^2 |pq| + (2/s)|pq| .
\] 
We have shown that the length of the path $\Pi$ is at most 
\[ \left( 2 (1+4/s)^2 + 2/s \right) |pq| , 
\] 
which is at most $( 2 + \epsilon/3 ) |pq|$ by our choice of $s$ in 
Algorithm~\ref{alg-improveB}. 

In the rest of the proof, we assume that $k > 2 \delta (\mu d +1)$.  
We define 
\[ m = k \bmod ( \delta ( \mu d + 1 ) ) 
\]
and 
\[ m' = \frac{k-m}{\delta (\mu d  + 1)} .
\] 
We consider the sequence of $c$-nodes 
\[ w=w_0 , w_{\delta(\mu d+1)+m} , w_{2\delta(\mu d +1)+m} ,  
      w_{3\delta(\mu d +1)+m} , \ldots , w_k=u , 
\] 
and define $\Pi$ to be the path 
\[  
\begin{array}{ccccccc} 
  p & \rightarrow & \rep(\cl(S_{w_0})) & \rightarrow & 
            \rep(S_{w_{\delta(\mu d +1)+m}}) & & \\
    & \rightarrow & \rep(\cl(S_{w_{2\delta(\mu d +1)+m }})) & 
            \rightarrow & \rep(S_{w_{2\delta(\mu d+1) +m}}) & & \\ 
    & \rightarrow & \rep(\cl(S_{w_{3\delta(\mu d +1)+m }})) & 
            \rightarrow & \rep(S_{w_{3\delta(\mu d+1) +m}}) & & \\ 
    & \vdots & & \vdots & & & \\ 
    & \rightarrow & \rep(\cl(S_{w_k})) & \rightarrow & 
            \rep(S_{w_k}) & = & \rep(S_u) . 
\end{array} 
\] 
It follows from Algorithm~\ref{alg-improveB} that $\Pi$ is a path in 
$G$. We will show that the length of this path is at most 
$(2 + \epsilon/3) |pq|$. 

We have shown already (see (\ref{paz11})) that the length of the first 
edge on $\Pi$ satisfies 
\[ | p , \rep(\cl(S_{w_0})) | \leq (1+4/s) \ell_0 . 
\] 
The length of the second edge satisfies 
\begin{eqnarray*} 
  | \rep(\cl(S_{w_0})) , \rep(S_{w_{\delta(\mu d +1)+m}}) | 
   & \leq &  
     |\rep(\cl(S_{w_0})) , p | + 
     | p , \rep(S_{w_{\delta(\mu d +1)+m}}) |  \\ 
   & \leq & 
     (1+4/s) \ell_0 + | p , \rep(S_{w_{\delta(\mu d +1)+m}}) | . 
\end{eqnarray*} 
Since $p$ and $\rep(S_{w_{\delta(\mu d +1)+m}})$ are both contained 
in $S_u$, it follows from Lemma~\ref{lemma-912} that 
\[ | p , \rep(S_{w_{\delta(\mu d +1)+m}}) | \leq (2/s) |pq| .
\] 
Thus, the length of the second edge on $\Pi$ satisfies 
\[ | \rep(\cl(S_{w_0})) , \rep(S_{w_{\delta(\mu d +1)+m}}) | 
    \leq (1+4/s) \ell_0 + (2/s) |pq| .
\] 
Let $2 \leq j \leq m'$. We have seen in (\ref{paz23}) in the proof of 
Lemma~\ref{lemma-levels} that the length of the edge 
\[ ( \rep(\cl(S_{w_{j \delta(\mu d +1)+m}})) , 
     \rep(S_{w_{j \delta(\mu d +1)+m}}) )
\]
satisfies 
\[ | \rep(\cl(S_{w_{j \delta(\mu d +1)+m}})) , 
     \rep(S_{w_{j \delta(\mu d +1)+m}}) | \leq 
     (1+4/s) \ell_{j \delta(\mu d +1)+m} . 
\] 
Again, let $2 \leq j \leq m'$. Since 
\[ \rep( S_{w_{(j-1) \delta(\mu d+1)+m}}) \in 
    S_{w_{j \delta(\mu d+1)+m}} \subseteq S_{x_{j \delta(\mu d+1)+m}} 
\]
and 
\[ \rep(\cl(S_{w_{j \delta(\mu d +1)+m}})) \in 
    S_{y_{j \delta(\mu d +1)+m}} 
\] 
it follows from Lemma~\ref{lemma-912} that the length of the edge 
\[ ( \rep( S_{w_{(j-1) \delta(\mu d+1)+m}}) , 
     \rep(\cl(S_{w_{j \delta(\mu d +1)+m}})) ) 
\] 
satisfies 
\[ | \rep( S_{w_{(j-1) \delta(\mu d+1)+m}}) , 
     \rep(\cl(S_{w_{j \delta(\mu d +1)+m}})) | \leq 
     (1+4/s) \ell_{j \delta(\mu d +1)+m} . 
\] 
We have shown that the length of $\Pi$ is at most     
\[ (2/s) |pq| + 
   2 ( 1 + 4/s ) 
   \left( \ell_0 + \sum_{j=2}^{m'} \ell_{j \delta(\mu d +1) +m} 
   \right) . 
\] 
The definition of $\ell_0,\ell_1,\ldots,\ell_k$ implies that this 
sequence is non-decreasing. Thus, 
$\ell_0 \leq \ell_{\delta(\mu d +1) +m}$ and it follows that the 
length of $\Pi$ is at most 
\[ (2/s) |pq| + 2(1+4/s) \sum_{j=1}^{m'} \ell_{j \delta(\mu d +1) +m} . 
\] 
Recall inequality (\ref{paz4}) in the proof of Lemma~\ref{lemma-levels}, 
which states that 
\[ \ell_i \leq \frac{1}{2} \ell_{i + \mu d + 1} . 
\] 
By applying this inequality $\delta$ times, we obtain 
\[ \ell_i \leq \left( \frac{1}{2} \right)^{\delta} 
                  \ell_{i + \delta(\mu d+1)} . 
\] 
For $i = j \delta(\mu d+1)+m$, this becomes  
\[ \ell_{j \delta(\mu d+1)+m} \leq \left( \frac{1}{2} \right)^{\delta} 
                  \ell_{(j+1) \delta(\mu d+1)+m} . 
\] 
By repeatedly applying this inequality, we obtain, for $h \geq j$,  
\[ \ell_{j \delta(\mu d+1)+m} \leq 
         \left( \frac{1}{2} \right)^{(h-j)\delta} 
                  \ell_{h \delta(\mu d+1)+m} . 
\] 
For $h = m'$, the latter inequality becomes 
\[ \ell_{j \delta(\mu d+1)+m} \leq 
         \left( \frac{1}{2} \right)^{(m'-j)\delta} \ell_k . 
\] 
It follows that  
\begin{eqnarray*} 
  \sum_{j=1}^{m'} \ell_{j \delta(\mu d +1) +m} & \leq & 
  \sum_{j=1}^{m'} \left( \frac{1}{2} \right)^{(m'-j)\delta} \ell_k \\ 
   & = & \sum_{i=0}^{m'-1} \left( \frac{1}{2} \right)^{i\delta} \ell_k  
              \\  
   & \leq & \sum_{i=0}^{\infty} 
                \left( \frac{1}{2^{\delta}} \right)^i \ell_k  \\ 
   & = & \frac{2^{\delta}}{2^{\delta}-1} \ell_k . 
\end{eqnarray*} 
According to (\ref{paz10}) in the proof of Lemma~\ref{lemma-levels}, 
we have 
\[ \ell_k \leq (1+4/s) |pq| . 
\] 
We have shown that the length of the path $\Pi$ is at most 
\[ \left( 2/s + 2 (1+4/s)^2 \, \frac{2^{\delta}}{2^{\delta}-1} 
   \right) |pq| . 
\] 
Our choices of $s$ and $\delta$ (see Algorithm~\ref{alg-improveB}) 
imply that $2/s \leq \epsilon/6$, $(1+4/s)^2 \leq 1 + \epsilon/36$ and  
$\frac{2^{\delta}}{2^{\delta}-1} \leq 1 + \epsilon/36$. 
Therefore, the length of $\Pi$ is at most 
\[ \left( \epsilon/6 + 2 ( 1 + \epsilon/36)^2 \right) |pq|  
     \leq ( 2 + \epsilon/3 ) |pq| , 
\]
where the latter inequality follows from our assumption that 
$0 < \epsilon < 1$. This completes the proof.  
\end{proof}

\begin{lemma} 
\label{paz7} 
Let $n = |S|$. The graph $G$ computed by Algorithm~\ref{alg-improveB} 
is a $(5+\epsilon)$-spanner of the complete $k$-partite graph 
$K_{C_1 \ldots C_k}$ and the number of edges of this graph is $O(n)$. 
The running time of Algorithm~\ref{alg-improveB} is $O(n \log n)$. 
\end{lemma}
\begin{proof} 
The proof for the upper bound on the stretch factor is similar to the 
one of Lemma~\ref{thm-pathS}. The difference is that instead of the 
value $t'$ that was used in the proof of Lemma~\ref{thm-pathS}, we now 
use the value $t'=2+\epsilon/3$ of Lemma~\ref{lem-improvB}. Thus, the 
stretch factor for the base case of the induction and for Case~1 is  
at most 
\[ (1 + 4/s) + 2 t' = 5 + 4/s + 2\epsilon/3 , 
\] 
which is at most $5+\epsilon$, because of our choice for $s$ in 
Algorithm~\ref{alg-improveB}. For Cases~2 and~3, the stretch factor is 
at most (see (\ref{paz6}) in the proof of Lemma~\ref{thm-pathS}, where 
$t = 5 + \epsilon$) 
\[ t' + (1 + 4/s) + 4t/s = 3 + \epsilon/3 + (4/s)(6+\epsilon) , 
\] 
which is at most $5+\epsilon$, again because of our choice for $s$. 
Finally, the stretch factor for Case~4 is at most (see (\ref{paz666}) 
in the proof of Lemma~\ref{thm-pathS}, where $t = 5 + \epsilon$) 
\[ (1+4/s) + 8t/s = 1 + (4/s) (11 + 2 \epsilon) , 
\] 
which is at most $5+\epsilon$, because of our choice for $s$. 

The analysis for the number of edges is the same as in 
Lemma~\ref{thm-numberOfEdges}, except that the number of edges that 
are added to each $c$-node in the modified for-loop is 
$2 \delta (\mu d +1)$ instead of one as is in 
Algorithm~\ref{alg-biSpan}. Finally, the analysis of the running time 
is the same as in Lemma~\ref{lemtime}. 
\end{proof}

We have proved the following result. 

\begin{theorem} 
Let $k \geq 2$ be an integer, let $S$ be a set of $n$ points in 
$\mathbb{R}^d$ which is partitioned into $k$ subsets 
$C_1,C_2,\ldots,C_k$, and let $0 < \epsilon < 1$ be a real constant. 
In $O(n \log n)$ time, we can compute a $(5+\epsilon)$-spanner of the 
complete $k$-partite graph $K_{C_1 \ldots C_k}$ having $O(n)$ edges. 
\end{theorem}

\section{Improving the Stretch Factor} 
\label{section-tightBounds} 
We have shown how to compute a $(5+\epsilon)$-spanner with $O(n)$ 
edges of any complete $k$-partite graph. In this section, we show that 
if we are willing to use $O(n \log n)$ edges, the stretch factor can 
be reduced to $3+\epsilon$. We start by showing that a stretch factor 
less than $3$, while using $O(n \log n)$ edges, is not possible.    

\begin{theorem} 
Let $c>0$ be a constant and let $n$ and $k$ be positive integers with 
$2 \leq k \leq n - 2c \sqrt{n \log n}$. For every real number 
$0 < \epsilon < 1$, there exists a complete $k$-partite geometric 
graph $K$ with $n$ vertices such that the following is true: 
If $G$ is any subgraph of $K$ with at most $c^2 n \log n$ edges,
then the stretch factor of $G$ is at least $3 - \epsilon$.   
\end{theorem}
\begin{proof} 
Let $D_1$, $D_2$, and $D_3$ be three disks of radius $\epsilon/12$ 
centered at the points $(0,0)$, $(1+\epsilon/6,0)$, and 
$(2+\epsilon/3,0)$, respectively. We place $(n-k+1)/2$ red 
points inside $D_1$ and $(n-k+1)/2$ blue points inside $D_2$. 
The remaining $k-2$ points are placed inside $D_3$ and each 
of these points has a distinct color (which is not red or blue). 
Let $K$ be the complete $k$-partite geometric graph defined by 
these $n$ points. We claim that $K$ satisfies the claim in 
the theorem. 

Let $G$ be an arbitrary subgraph of $K$ and assume that $G$ contains 
at most $c^2 n \log n$ edges. We will show that the stretch factor of 
$G$ is at least $3 - \epsilon$.   

Assume that $G$ contains all red-blue edges. Then the number of edges 
in $G$ is at least $( (n-k+2) / 2 )^2$. Since 
$k \leq n - 2c \sqrt{n \log n}$, this quantity is larger 
than $c^2 n \log n$. Thus, there is a red point $r$ and a blue point 
$b$, such that $(r,b)$ is not an edge in $G$. The length of a shortest 
path in $G$ between $r$ and $b$ is at least $3$. Since 
$|rb| \leq 1 + \epsilon / 3$, it follows that the stretch factor of $G$ 
is at least $\frac{3}{1+ \epsilon/3}$, which is at least $3-\epsilon$. 
\end{proof}

\begin{theorem}
Let $k \geq 2$ be an integer, let $S$ be a set of $n$ points in 
$\mathbb{R}^d$ which is partitioned into $k$ subsets 
$C_1,C_2,\ldots,C_k$, and let $0 < \epsilon < 1$ be a real constant. 
In $O(n \log n)$ time, we can compute a $(3+\epsilon)$-spanner of the 
complete $k$-partite graph $K_{C_1 \ldots C_k}$ having $O(n \log n)$ 
edges. 
\end{theorem}
\begin{proof} 
Consider the following variant of the WSPD. For every pair $\{X,Y\}$ 
in the standard WSPD, where $|X| \leq |Y|$, we replace this pair 
by the $|X|$ pairs $\{\{x\} ,B \}$, where $x$ ranges over all points 
of $X$. Thus, in this new WSPD, each pair contains at least one 
singleton set. Callahan and Kosaraju~\cite{callahan95} showed that this 
new WSPD consists of $O(n \log n)$ pairs.

We run Algorithm~\ref{alg-improveB} on the set $S$, using this new 
WSPD. Let $G$ be the graph that is computed by this algorithm. 
Observe that Lemma~\ref{lem-improvB} still holds for $G$. 
In the proof of Lemma~\ref{paz7} of the upper bound on the stretch 
factor of $G$, we have to apply Lemma~\ref{lem-improvB} only once. 
Therefore, the stretch factor of $G$ is at most $3+\epsilon$. 
\end{proof}

\section{Conclusion}\label{section-bi-span-conclusion}
We have shown that for every complete $k$-partite geometric graph $K$  
in $\mathbb{R}^d$ with $n$ vertices and for every constant 
$\epsilon > 0$, 
\begin{enumerate} 
\item a $(5+\epsilon)$-spanner of $K$ having $O(n)$ edges can be 
      computed in $O(n \log n)$ time, 
\item a $(3+\epsilon)$-spanner of $K$ having $O(n \log n)$ edges can be 
      computed in $O(n \log n)$ time. 
\end{enumerate} 
The latter result is optimal for 
$2 \leq k \leq n - \Theta(\sqrt{n \log n} )$, because a spanner of 
$K$ having stretch factor less than $3$ and having $O(n \log n)$ edges 
does not exist for all complete $k$-partite geometric graphs.  

We leave open the problem of determining the best stretch factor that 
can be obtained by using $O(n)$ edges. 
  
Future work may include verifying other properties that are known for 
the general geometric spanner problem. For example, is there a spanner 
of a complete $k$-partite geometric graph that has bounded degree? Is 
there a spanner of a complete $k$-partite geometric graph that is 
planar? From a more general point of view, it seems that little is
known about geometric spanners of graphs other than the complete
graph. The unit disk graph received great attention, but there are a
large family of other graphs that also deserve attention.

\bibliographystyle{plain}
\bibliography{k-partite-spanners}

\end{document}